\documentclass[smallcondensed]{svjour3}     

\def\PubType{Conference}

\usepackage{graphicx}
\usepackage{amssymb}
\usepackage[cmex10]{amsmath}
\usepackage{color}
\usepackage{theorem}
\usepackage{hyperref}
\hypersetup{colorlinks=false,pdfborderstyle={/S/U/W 1},pdfborder=0 0 1}
\usepackage{url}
\usepackage{breakurl}
\usepackage{verbatim}
\usepackage{ifthen}
\usepackage[ruled,vlined,titlenumbered]{algorithm2e}

\smartqed 

\newcommand{\F}[1]{\ensuremath{\mathbb{F}_{#1}}}
\newcommand{\Fq}{\F{q}}

\newcommand{\Fx}[1]{\ensuremath{\F{#1}[X]}}
\newcommand{\Fqx}{\Fx{q}}
\newcommand{\Fxy}[1]{\ensuremath{\F{#1}[X,Y]}}
\newcommand{\Fqxy}{\Fxy{q}}
\newcommand{\wdeg}[2]{\operatorname{wdeg}_{#1,#2}}
\newcommand{\yC}[2]{#1_{[#2]}} 
\newcommand{\RS}[2]{\ensuremath{\mathcal{GRS}(#1,#2)}}

\newcommand{\weight}[1]{\operatorname{weight}(#1)}
\newcommand{\defeq}{\triangleq}
\newcommand{\half}{\tfrac 1 2}
\newcommand{\diag}{\mathrm{diag}} 
\newcommand{\LT}[1]{\textnormal{\footnotesize LT}(#1)}
\newcommand{\LP}[1]{\textnormal{\footnotesize LP}(#1)}
\newcommand{\OD}[1]{\operatorname{D}(#1)}
\newcommand{\ZZ}{\mathbb Z}
\newcommand{\NN}{\mathbb N}
\newcommand{\M}[2][\empty]{
  \ifthenelse{\equal{#1}{\empty}}
    {\ensuremath{\mathcal{#2}}}
    {\ensuremath{\mathcal{#2}_{#1}}}
}
\newcommand{\Mod}[2][\empty]{
  \ifthenelse{\equal{#1}{\empty}}
    {\ensuremath{#2}}
    {\ensuremath{#2_{#1}}}
}

\newcommand{\T}[1]{^{(#1)}} 
\newcommand{\sell}[1][s,\ell]{{#1}}  
\newcommand{\sellI}{\sell[s,\ell+1]}
\newcommand{\sellII}{\sell[s+1,\ell+1]}
\newcommand{\stepI}{^\mathrm I}
\newcommand{\stepII}{^\mathrm{II}}

\def\intermediator#1{\hat{#1}}
\def\is{\intermediator s}
\def\iell{\intermediator \ell}
\def\itau{\intermediator \tau}
\def\iQ{\intermediator Q}
\newcommand{\isell}[1][\is,\iell]{\sell[#1]}
\def\var#1{\textup{\textsf{#1}}}
\newcommand{\printalgos}[1]{\begin{minipage}{0.9\textwidth}\begin{center}\begin{algorithm}[H]#1
\end{algorithm}\end{center}\end{minipage}}

\def\confType{Conference} 
\def\jourType{Journal}

\definecolor{orange}{rgb}{1,0.5,0}
\ifx\PubType\confType
  \newcommand{\todo}[1]{}

\else\ifx\PubType\jourType
  \newcommand{\todo}[1]{}

\else
  \newcommand{\todo}[1]{{\color{red}[#1]}}

\fi\fi
\newtheorem{thm}{Theorem}
\newtheorem{cor}[thm]{Corollary}
\newtheorem{lem}[thm]{Lemma}
\newtheorem{exa}[thm]{Example} 
\newtheorem{defi}[thm]{Definition} 
\renewcommand{\tilde}{\widetilde}

\begin{document}

\def\thetitle{Multi-Trial Guruswami--Sudan Decoding for Generalised Reed--Solomon Codes}
\title{\thetitle}
\titlerunning{\thetitle}

\author{Johan S.~R. Nielsen \and Alexander Zeh}
\authorrunning{J.~S.~R.~Nielsen \and A.~Zeh}
\institute{J.~S.~R.~Nielsen is with the Institute of Mathematics, Technical University of Denmark \thanks{TODO: This shows up in any place?} \\ 
\email{\texttt{j.s.r.nielsen@mat.dtu.dk}} \\
A.~Zeh is with the Institute of Communications Engineering, University of Ulm, Germany \thanks{This work was supported by the German Research Council "Deutsche Forschungsgemeinschaft" (DFG) under Grants Bo867/22-1.} and Research Center INRIA Saclay - \^{I}le-de-France, \'{E}cole Polytechnique, France\\
\email{\texttt{alex@codingtheory.eu}}
}

\date{\today}
\maketitle

\begin{abstract}
An iterated refinement procedure for the Guruswami--Sudan list decoding algorithm for Generalised Reed--Solomon codes based on Alekhnovich's module minimisation is proposed.
The method is parametrisable and allows variants of the usual list decoding approach.
In particular, finding the list of \emph{closest} codewords within an intermediate radius can be performed with improved average-case complexity while retaining the worst-case complexity.
\end{abstract}


\section{Introduction}
\todo{/Johan: We have very un-systematic use of cdot in matrix multiplications. I vote for removing all of them}

Since the discovery of a polynomial-time hard-decision list decoder for Generalised Reed--Solomon (GRS) codes by Guruswami and Sudan (GS)~\cite{sudan97,guruSudan99} in the late 1990s, much work has been done to speed up the two main parts of the algorithm: interpolation and root-finding.
Notably, for interpolation Beelen and Brander~\cite{beelenBrander} mixed the module reduction approach by Lee and O'Sullivan~\cite{leeOSullivan08} with the parametrisation of Zeh~\textit{et al.}~\cite{zeh11interpol}, and employed the fast module reduction algorithm by Alekhnovich~\cite{alekhnovich05}.
Bernstein~\cite{bernstein11} pointed out that a slightly faster variant can be achieved by using the reduction algorithm by Giorgi~\textit{et al.}~\cite{giorgi03}.

For the root-finding step, one can employ the method of Roth and Ruckenstein~\cite{rothRuckenstein00} in a divide-and-conquer fashion, as described by Alekhnovich~\cite{alekhnovich05}.
This step then becomes an order of magnitude faster than interpolation, leaving the latter as the main target for further optimisations.

For a given code, the GS algorithm has two parameters, both positive integers: the interpolation multiplicity $s$ and the list size $\ell$.
Together with the code parameters they determine the decoding radius $\tau$. To achieve a higher decoding radii for some given GRS code, one needs higher $s$ and $\ell$, and the value of these strongly influence the running time of the algorithm.

In this work, we present a novel iterative method: we first solve the interpolation problem for $s = \ell = 1$ and then iteratively refine this solution for increasing $s$ and $\ell$. In each step of our algorithm, we obtain a valid solution to the interpolation problem for these intermediate parameters.
The method builds upon that of Beelen--Brander~\cite{beelenBrander} and has the same asymptotic complexity.

\def\inter#1{\hat #1}
The method therefore allows a fast \emph{multi-trial} list decoder when our aim is just to find the list of codewords with minimal distance to the received word.
At any time during the refinement process, we will have an interpolation polynomial for intermediate parameters $\inter s \leq s$, $\inter \ell \leq \ell$ yielding an intermediate decoding radius $\inter \tau\leq \tau$.
If we perform the root-finding step of the GS algorithm on this, all codewords with distance at most $\inter \tau$ from the received are returned; if there are any such words, we break computation and return those; otherwise we continue the refinement.
We can choose any number of these trials, e.g. for each possible intermediate decoding radius between half the minimum distance and the target $\tau$.

Since the root-finding step of GS is cheaper than the interpolation step, this multi-trial decoder will have the same asymptotic worst-case complexity as the usual GS using the Beelen--Brander interpolation; however, the average-case complexity is better since fewer errors are more probable.

This contribution is structured as follows.
In the next section we give necessary preliminaries and state the GS interpolation problem for decoding GRS codes.
In Section~\ref{sec_modulemini} we give a definition and properties of minimal matrices.
Alekhnovich's algorithm can bring matrices to this form, and we give a more fine-grained bound on its asymptotic complexity.
Our new iterative procedure is explained in detail in Section~\ref{sec_iterative}.

\section{Preliminaries} \label{sec_prelim}
\subsection{Notation}\label{ssec_notation}
Let $\Fq$ be the finite field of order $q$ and let $\Fqx$ be the polynomial ring over 
$\Fq$ with indeterminate $X$. Let $\Fqxy$ denote the polynomial ring in the variables $X$ and $Y$ 
and let $\wdeg{u}{v} X^iY^j \defeq ui + vj$ be the $(u,v)$-weighted degree of $X^iY^j$.

A vector of length $n$ is denoted by $\vec{v} = (v_0, \dots, v_{n-1})$.
If $\vec{v}$ is a vector over $\Fqx$, let $\deg \vec{v} \defeq \max_i\{ \deg v_i(X) \}$.
We introduce the leading position as $\LP{\vec{v}}= \max_i\{ i | \deg v_i(X) = \deg \vec{v} \}$ and the leading term $\LT{\vec{v}} = v_{\LP{\vec{v}}}$ is the term at this position.
An $m \times n$ matrix is denoted by $\M{V}=\| v_{i,j}\|_{i=0,j=0}^{m-1,n-1}$.
The rows of such a matrix will be denoted by lower-case letters, e.g. $\vec v_0, \ldots, \vec v_{m-1}$.
Furthermore, let $\deg \M{V} = \sum_{i=0}^{m-1} \deg \vec{v}_i$.
Modules are denoted by capital letters such as $M$.

\subsection{Interpolation-Based Decoding of GRS Codes} \label{ssec_ipdecoding}
Let $\alpha_0, \dots, \alpha_{n-1}$ be $n$ nonzero distinct elements 
of $\Fq$ with $n < q$ and let $w_0,\dots,w_{n-1}$ be $n$ (not necessarily distinct) nonzero elements of $\Fq$.
A GRS code \RS{n}{k} of length $n$ and dimension $k$ over $\Fq$ is given by
\begin{align} \label{eq_defRScode}
  \RS{n}{k} \defeq \big\{ (w_0 f(\alpha_{0}), \dots, w_{n-1}f(\alpha_{n-1})) : f(X) \in \Fqx,\, \deg f(X) < k \big\}.
\end{align}
GRS codes are \emph{Maximum Distance Separable} (MDS) codes, i.e., their minimum Hamming distance is $d=n-k+1$.
We shortly explain the interpolation problem of GS~\cite{guruSudan99,sudan97} for decoding GRS codes in the following.
\begin{thm}[Guruswami--Sudan for GRS Codes~\cite{guruSudan99,sudan97}] \label{thm_GSproblem}
Let $\vec{c} \in \RS{n}{k}$ be a codeword and $f(X)$ the corresponding information polynomial as defined in~\eqref{eq_defRScode}. Let $\vec{r} = (r_0,\dots,r_{n-1}) = \vec{c} + \vec{e}$ be a received word where $\weight{\vec e} \leq \tau$. Let $r_i^{\prime}$ denote $r_{i}/w_{i}$.

Let $Q(X,Y) \in \Fqxy$ be a nonzero polynomial that passes through the $n$ points $(\alpha_0,r_0^{\prime}),$ $ \dots, (\alpha_{n-1},r_{n-1}^{\prime})$ with multiplicity $s \geq 1$, has $Y$-degree at most $\ell$, and $\wdeg{1}{k-1} Q(X,Y) < s(n-\tau)$.
Then $(Y-f(X)) \mid Q(X,Y)$.
\end{thm}
One can easily show that a polynomial $Q(X,Y)$ that fulfils the above conditions can be constructed whenever $E(s,\ell,\tau) > 0$, where
\begin{equation}\label{eqn_Edef}
  E(s,\ell,\tau) \defeq (\ell+1)s(n-\tau) - \tbinom{\ell+1} 2 (k-1) - \tbinom{s+1} 2 n
\end{equation}
is the difference between the maximal number of coefficients of $Q(X,Y)$, and the number of homogeneous linear equations on $Q(X,Y)$ specified by the interpolation constraint.
This determines the maximal number of correctable errors, and one can show that satisfactory $s$ and $\ell$ can always be chosen whenever $\tau < n-\sqrt{n(k-1)}$ (for $n\rightarrow \infty$ see e.g.~\cite{guruSudan99}).
\begin{defi}[Permissible Triples] \label{def_permiss}
An integer triple $(s,\ell,\tau) \in (\ZZ_+)^3$ is \emph{permissible} if $E(s,\ell,\tau) > 0$.
\end{defi}
We define also the decoding radius-function $\tau(s,\ell)$ as the greatest integer such that $(s,\ell,\tau(s,\ell))$ is permissible.

It is easy to show that $E(s,\ell,\tau) > 0$ for $s > \ell$ implies $\tau < \lfloor \frac {n-k} 2 \rfloor$, which is half the minimum distance.
Therefore, it never makes sense to consider $s > \ell$, and in the remainder we will always assume $s \leq \ell$.
Furthermore, we will also assume $s, \ell \in O(n^2)$ since this e.g.~holds for any $\tau$ for the closed-form expressions in \cite{guruSudan99}.

\begin{inJournal}
Let us illustrate the above. The following will be a running example throughout the article.
\begin{exa}
  \label{ex_rs164param}
  A $\RS{16}{4}$ code over \F{17} can uniquely correct $\tau_0 = (n-k)/2=6$ errors; unique decoding corresponds to $s_0 = \ell_0 = 1$ and one can confirm that $E(1, 1, 6) > 0$. 
  To attain a decoding radius $\tau_1 = 7$, one can choose $s_1 = 1$ and $\ell_1 = 2$ in order to obtain a permissible triple.
  Also $(1, 3, 7)$ is permissible, though less interesting since it does not give improved decoding radius.
  However, one finds e.g. $(2, 4, 8)$ and $(28, 64, 9)$ are permissible.
  For greater decoding radii, no triples are permissible.
\end{exa}
\end{inJournal}
\begin{inJournal}
\todo{The following is not true}
The following lemma is useful for our multi-trial decoding procedure of Section~\ref{sec_iterative}.
\begin{lem}
  \label{lem_permSdiff}
  If $(s,\ell,\tau)$ and $(\is, \iell, \itau)$ are both permissible triples as in Definition~\ref{def_permiss} such that $\itau > \tau$ and $\is - s > \iell - \ell$, then also $(s + (\iell-\ell), \iell, \itau)$ is permissible.
\end{lem}
\begin{proof}
  \def\D{\Delta}
  Let $\D = (\is-s) - (\iell-\ell) > 0$. Then
  \[
    E(s + (\iell-\ell), \iell, \itau) = E(\is - \D, \iell, \itau)
    = (\is - \D)(\ell+1)(n-\itau) - \tbinom {\iell+1} 2 (k-1) - \tbinom {\is - \D} 2 n.
  \]
  We want to show that this is greater than 0.
  One can easily verify that for $a,b \in \ZZ_+$ with $a > b$ then $\binom {a-b} 2 = \binom {a+1} 2 + \binom b 2 - ab$.
  Thus the above equals \todo{This is not true since the equation is unfixed}
  \begin{align*}
    & \is(n-\itau) - \tbinom {\iell+1} 2 (k-1) - \tbinom{\is+1} 2 n - \D(n-\itau) - \tbinom \D 2 n + \is n\D
    \\& = 
    E(\is,\iell,\itau) + \D\big(\itau + n(\is - \half(\D-1) - 1 ) \big).
  \end{align*}
  Since by assumption $E(\is,\iell,\itau) > 0$ and $\is - \half\D \geq s + \half\D > \half$, the above expression is greater than 0.
  \qed
\end{proof}
\end{inJournal}

\subsection{Module Reformulation of Guruswami--Sudan} \label{ssec_modreform}
Let $\Mod[\sell]{M} \subset \Fqxy$ denote the space of all bivariate polynomials passing through the points $(\alpha_0,r_0^{\prime}),\ldots,(\alpha_{n-1},r_{n-1}^{\prime})$ with multiplicity $s$ and with $Y$-degree at most $\ell$.
We are searching for an element of $\Mod[\sell]{M}$ with low $(1,k-1)$-weighted degree.

Following the ideas of Lee and O'Sullivan~\cite{leeOSullivan08}, we can first remark that $\Mod[\sell]{M}$ is an $\Fqx$ module.
Second, we can give an explicit basis for $\Mod[\sell]{M}$.
Define first two polynomials $G(X) = \prod_{i=0}^{n-1}(X-\alpha_i)$ as well as $R(X)$ as the Lagrange polynomial going through the points $(\alpha_i, r_i^\prime)$ for $i = 0, \ldots, n-1$.
Denote by $\yC Q t(X)$ the $Y^t$-coefficient of $Q(X,Y)$ when $Q$ is regarded over $\Fqx[Y]$.
\begin{lem}\label{lem_QdivG}
  Let $Q(X,Y) \in \Mod[\sell]{M}$. Then $G(X)^{s-t} \mid \yC Q t(X)$ for $t < s$.
\end{lem}
\begin{proof}
  $Q(X,Y)$ interpolates the $n$ points $(\alpha_i, r_i^{\prime})$ with multiplicity $s$, so for any $i$, $Q(X+\alpha_i, Y+r_i^{\prime}) = \sum_{j=0}^t \yC Q j(X+\alpha_j)(Y+r_j^{\prime})^j$ has no monomials of total degree less than $s$.
  Multiplying out the $(Y+r_j^{\prime})^j$-terms, $\yC Q t(X+\alpha_j)Y^t$ will be the only term with $Y$-degree $t$.
  Therefore $\yC Q t (X+\alpha_j)$ can have no monomials of degree less than $s-t$, which implies $(X-\alpha_i) \mid \yC Q t(X)$.
  As this holds for any $i$, we proved the lemma.
  \qed
\end{proof}

\begin{thm}
  \label{thm_Mbasis}
  The module $\Mod[\sell]{M}$ is generated as an $\Fqx$-module by the $\ell+1$ polynomials $P\T i(X,Y) \in \Fqxy$ given by
  \begin{align*}
    P\T t(X,Y) &= G(X)^{s-t}(Y-R(X))^t,         && \textrm{for } 0 \leq t < s,\\
    P\T t(X,Y) &= Y^{t-s}(Y-R(X))^s,            && \textrm{for } s \leq t \leq \ell.
  \end{align*}
\end{thm}
\begin{proof}
  It is easy to see that each $P\T t(X,Y) \in \Mod[\sell]{M}$ since both $G(X)$ and $(Y-R(X))$ go through the $n$ points $(\alpha_i, r_i^\prime)$ with multiplicity one, and that $G(X)$ and $(Y-R(X))$ divide $P\T t (X,Y)$ with total power $s$ for each $t$.

  To see that any element of $\Mod[\sell]{M}$ can be written as an $\Fqx$-combination of the $P\T t(X,Y)$, let $Q(X,Y)$ be some element of $\Mod[\sell]{M}$.
  Then the polynomial $Q\T{\ell-1}(X,Y) = Q(X,Y) - \yC Q \ell P\T \ell(X,Y) $ has $Y$-degree at most $\ell-1$. Since both $Q(X,Y)$ and $P\T \ell(X,Y)$ are in $\Mod[\sell]{M}$, so must $Q\T{\ell-1}(X,Y)$ be in $\Mod[\sell]{M}$.
  Since $P\T t(X,Y)$ has $Y$-degree $t$ and $\yC{P\T t} t(X) = 1$ for $t = \ell, \ell-1, \ldots, s$, we can continue reducing this way until we reach a $Q\T{s-1}(X,Y) \in \Mod[\sell]{M}$ with $Y$-degree at most $s-1$.
  From then on, we have $\yC{P\T t} t(X) = G(X)^{s-t}$, but by Lemma~\ref{lem_QdivG}, we must also have $G(X) \mid \yC{Q\T{s-1}}{s-1}(X)$, so we can also reduce by $P\T {s-1}(X,Y)$.
  This can be continued with the remaining $P\T t(X,Y)$, eventually reducing the remainder to 0.
  \qed
\end{proof}
\begin{inJournal}
We can represent the basis of $\Mod[\sell]{M}$ by the $(\ell+1) \times (\ell+1)$ matrix $\M[\sell]{A}=\| P^{(i)}_{[j]}(X,Y)\|_{i=0,j=0}^{\ell, \ell}$ over $\Fqx$, more explicitly we have:
\begin{equation} \label{eq_BasisOfModul}
\M[\sell]{A} \defeq \left[\begin{array}{cccccc}
        G^s & 0 & 0 & \dots & 0 & 0 \\
        G^{s-1}(-R)& G^{s-1} & 0 & \dots & 0& 0 \\
        & & \vdots & & \\
        (-R)^s & \binom{s}{1}(-R)^{s-1} & \binom{s}{2}(-R)^{s-2} & \dots & 0 & 0 \\
        0 & (-R)^{s} & \binom{s}{1}(-R)^{s-1} & \dots & 0 &0 \\
        & & \vdots & & \\
        0 & \dots & (-R)^{s} & \ldots & \binom{s}{s-1}(-R) & 1 
    \end{array}\right].
\end{equation}
\end{inJournal}
\begin{inConference}
We can represent the basis of $\Mod[\sell]{M}$ by the $(\ell+1) \times (\ell+1)$ matrix $\M[\sell]{A}=\| P^{(i)}_{[j]}(X,Y)\|_{i=0,j=0}^{\ell, \ell}$ over $\Fqx$.
\end{inConference}
Any $\Fqx$-linear combination of rows of $\M[\sell] A$ thus corresponds to an element in $\Mod[\sell] M$ by its $t$th term being the $\Fqx$-coefficient to $Y^t$.
All other bases of $\Mod[\sell]{M}$ can be similarly represented by matrices, and these will be unimodular equivalent to $\M[\sell]{A}$, i.e., they can be obtained by multiplying $\M[\sell] A$ on the left with an invertible matrix over $\Fqx$.

Extending the work of Lee and O'Sullivan~\cite{leeOSullivan08}, Beelen and Brander~\cite{beelenBrander} gave a fast algorithm for computing a satisfactory $Q(X,Y)$: start with $\M[\sell]{A}$
as a basis of $\Mod[\sell] M$ and compute a different, ``minimal'' basis of $\Mod[\sell]{M}$ where an element of minimal $(1, k-1)$-weighted degree appears directly.\footnote{%
Actually, in both \cite{leeOSullivan08,beelenBrander}, a slight variant of $\M[\sell] A$ is used, but the difference is non-essential.}

In the following section,  we give further details on how to compute such a basis, but our ultimate aims in Section \ref{sec_iterative} are different: we will use a minimal basis of $\Mod[\sell] M$ to efficiently compute one for $\Mod[\hat s,\hat \ell] M$ for $\hat s \geq s$ and $\hat \ell > \ell$.
This will allow an iterative refinement for increasing $s$ and $\ell$, where after each step we have such a minimal basis for $\Mod[\sell] M$. We then exploit this added flexibility in our multi-trial algorithm.

\section{Module Minimisation} \label{sec_modulemini}

Given a basis of $\Mod[\sell] M$, e.g.~$\M[\sell] A$, the module minimisation here refers to the process of obtaining a new basis, which is the smallest among all bases of $\Mod[\sell] M$ in a precise sense.
We will define this and connect various known properties of such matrices, and use this to more precisely bound the asymptotic complexity with which they can be computed by Alekhnovich's algorithm.
\begin{defi}[Weak Popov Form \protect{\cite{mulders03}}]
  \label{def_weakpopov}
  A matrix $\M V$ over $\Fqx$ is in \emph{weak Popov} form if an only if the leading position of each row is different.
\end{defi}
We are essentially interested in short vectors in a module, and the following lemma shows that the simple concept of weak Popov form will provide this.
It is a paraphrasing of \cite[Proposition 2.3]{alekhnovich05} and we omit the proof.
\begin{lem}[Minimal Degree]
  \label{lem_minrow}
  If a square matrix $\M V$ over $\Fqx$ is in weak Popov form, then one of its rows has minimal degree of all vectors in the row space of $\M V$.
\end{lem}
Denote now by $\Mod[\ell]{\M W}$ the diagonal $(\ell+1) \times (\ell+1)$ matrix over $\Fqx$:
\begin{equation} \label{eq_DiagWeight}
  \Mod[\ell]{\M W} \defeq \diag \left(1,X^{k-1}, \dots, X^{\ell(k-1)} \right).
\end{equation}
Since we seek an element of minimal $(1,k-1)$-weighted degree, we also need the following corollary.
\begin{cor}[Minimal Weighted Degree] \label{cor_WeightSol}
  Let $\M B \in \Fqx^{(\ell+1) \times (\ell+1)}$ be the matrix representation of a basis of $\Mod[s,\ell] M$.
  If $\M B  \Mod[\ell]{\M W}$ is in weak Popov form, then one of the rows of $\M B$ corresponds to a polynomial in $\Mod[s,\ell] M$ with minimal $(1,k-1)$-weighted degree.
\end{cor}
\begin{proof}
  Let $\M{\tilde B} = \M B  \Mod[\ell]{\M W}$.
  Now, $\M{\tilde B}$ will correspond to the basis of an $\Fqx$-module $\tilde M$ isomorphic to $\Mod[s,\ell] M$, where an element $Q(X,Y) \in \Mod[s,\ell] M$ is mapped to $Q(X,X^{k-1}Y) \in \tilde M$.
  By Lemma \ref{lem_minrow}, the row of minimal degree in $\M{\tilde B}$ will correspond to an element of $\tilde M$ with minimal $X$-degree.
  Therefore, the same row of $\M B$ corresponds to an element of $\Mod[s,\ell] M$ with minimal $(1,k-1)$-weighted degree.
  \qed
\end{proof}

%
We introduce what will turn out to be a measure of how far a matrix is from being in weak Popov form.

\begin{defi}[Orthogonality Defect \protect{\cite{lenstra85}}]
  \label{def_orthogonality}
  Let the \emph{orthogonality defect} of a square matrix $\M V$ over $\Fqx$ be defined as $\OD{\M V} \defeq \deg \M V - \deg \det \M V$.
\end{defi}

\begin{lem}\label{lem_popovreduces}
  If a square matrix $\M V$ over $\Fqx$ is in weak Popov form then $\OD{\M V} = 0$.
\end{lem}
\begin{proof}
  Let $\vec{v}_0,\ldots,\vec{v}_{m-1}$ be the rows of $\M{V} \in \Fqx^{m \times m}$ and $v_{i,0},\ldots, v_{i,m-1}$ the elements of $\vec v_i$.
  In the alternating sum-expression for $\det \M V$, the term $\prod_{i=0}^{m-1} \LT{\vec{v}_i}$ will occur since the leading positions of $\vec{v}_i$ are all different.
  Thus $\deg \det \M V = \sum_{i=0}^{m-1} \deg \LT{\vec{v}_i} = \deg \M V$ unless leading term cancellation occurs in the determinant expression.
  However, no other term in the determinant has this degree: regard some (unsigned) term in $\det \M V$, say $t = \prod_{i=0}^{m-1} v_{i, {\sigma(i)}}$ for some permutation $\sigma \in S_m$.
  If not $\sigma(i) = \LP{\vec{v}_i}$ for all $i$, then there must be an $i$ such that $\sigma(i) > \LP{\vec{v}_i}$ since $\sum_j \sigma(j)$ is the same for all $\sigma \in S_m$.
  Thus, $\deg v_{i, {\sigma(i)}} < \deg v_{i, \LP{\vec{v}_i}}$.
  As none of the other terms in $t$ can have greater degree than their corresponding row's leading term, we get $\deg t < \sum_{i=0}^{m-1} \deg \LT{\vec{v}_i}$.
  Thus, $\OD{\M V} = 0$.
  However, the above also proves that the orthogonality defect is at least 0 for \emph{any} matrix.
  Since any matrix unimodular equivalent to $\M V$ has the same determinant, $\M V$ must therefore have minimal row-degree among these matrices.
  \qed
\end{proof}
\begin{inJournal}
Let us consider an example to illustrate all the above.
\begin{exa}[Orthogonality Defect and Weak-Popov Form] \label{ex_weakpopov}
Let us consider the following matrices $\M{V}_0, \dots, \M{V}_3 \in \Fx{2}^{3 \times 3}$. From matrix $\M{V}_i$ to $\M{V}_{i+1}$ we performed one row-operation to reduce the leading term:
\begin{scriptsize}
\begin{equation*}
\M{V}_0 =  \left[\begin{array}{ccc}
        1 & X^2 & X \\
        0 & X^3 & X^2 \\
        X & 1 & 0 
    \end{array}\right]
\xrightarrow{(1,2)}
\M{V}_1 =  \left[\begin{array}{ccc}
        1 & X^2 & X \\
        X & 0 & 0 \\
        X & 1 & 0 
    \end{array}\right]
\xrightarrow{(3,2)}
\M{V}_2 =  \left[\begin{array}{ccc}
        1 & X^2 & X \\
        X & 0 & 0 \\
        0 & 1 & 0 
    \end{array}\right]
\xrightarrow{(1,3)}
\M{V}_3 =  \left[\begin{array}{ccc}
        1 & 0 & X \\
        X & 0 & 0 \\
        0 & 1 & 0 
    \end{array}\right],
\end{equation*}
\end{scriptsize}where the indexes $(i_1,i_2)$ on the arrow indicated the concerned rows.
The orthogonality defect is decreasing; $\OD{\M{V}_0} = 3 \rightarrow \OD{\M{V}_1} = 2 \rightarrow \OD{\M{V}_2} = 1 \rightarrow \OD{\M{V}_3} = 0$.
\end{exa}
\end{inJournal}

Alekhnovich~\cite{alekhnovich05} gave a fast algorithm for transforming a matrix over $\Fqx$ to weak Popov form.
For the special case of square matrices, a finer description of its asymptotic complexity can be reached in terms of the orthogonality defect, and this is essential for our decoder.
\begin{lem}[Alekhnovich's Row-Reducing Algorithm] \label{lem_alekcompl}
Alekhnovich's algorithm inputs a matrix $\M V \in \Fqx^{m\times m}$ and outputs a unimodular equivalent matrix which is in weak Popov form.
  Let $N$ be the greatest degree of a term in $\M V$.
  If $N \in O(\OD{\M V})$ then the algorithm has asymptotic complexity:
  \begin{equation*}
  O\big( m^3 \OD{\M V}\log^2 \OD{\M V} \log\log \OD{\M V} \big) \quad \text{operations over $\Fq$}.
  \end{equation*}
\end{lem}
\begin{proof}
  The description of the algorithm as well as proof of its correctness can be found in~\cite{alekhnovich05}.
  We only prove the claim on the complexity.
  The method $R(\M V, t)$ of \cite{alekhnovich05} computes a unimodular matrix $\M U$ such that $\deg(\M U\M V) \leq \deg \M V - t$ or $\M U\M V$ is in weak Popov form.
  According to \cite[Lemma 2.10]{alekhnovich05}, the asymptotic complexity of this computation is  in $O(m^3 t \log^2 t \log\log t)$.
  Due to Lemma~\ref{lem_popovreduces}, we can set $t = \OD{\M V}$ to be sure that $\M U\M V$ is in weak Popov form.
  What remains is just to compute the product $\M U\M V$.
  Due to \cite[Lemma 2.8]{alekhnovich05}, each entry in $\M U$ can be represented as $p(X)X^d$ for some $d \in \NN_0$ and $p(X) \in \Fqx$ of degree at most $2t$.
  If therefore $N \in O(\OD{\M V})$, the complexity of performing the matrix multiplication using the naive algorithm is $O(m^3\OD{\M V})$.
  \qed
\end{proof}

\section{Multi-Trial List Decoding}
\label{sec_iterative}

\subsection{Basic Idea}
\label{ssec_outline}

Using the results of the preceding section, we show in Section~\ref{ssec_steptypeone} that given a basis of $\Mod[\sell] M$ as a matrix $\M[\sell] B$ in weak Popov form, then we can write down a matrix $\M[s,\ell+1] C\stepI$ which is a basis of $\Mod[s,\ell+1] M$ and whose orthogonality defect is much lower than that of $\M[s,\ell+1] A$.
This means that reducing $\M[s,\ell+1] C\stepI$ to weak Popov form using Alekhnovich's algorithm is faster than reducing $\M[s,\ell+1] A$.
We call this kind of refinement a ``micro-step of type I''.
In Section~\ref{ssec_steptypetwo}, we similarly give a way to refine a basis of $\Mod[\sell] M$ to one of $\Mod[s+1,\ell+1] M$, and we call this a micro-step of type II.

If we first compute a basis in weak Popov form of $\Mod[1,1] M$ using $\M[1,1] A$, we can perform a sequence of micro-steps of type I and II to compute a basis in weak Popov form of $\Mod[s,\ell] M$ for any $s,\ell$ with $\ell \geq s$.
After any step, having some intermediate $\is \leq s$, $\iell \leq \ell$, we will thus have a basis of $\Mod[\isell] M$ in weak Popov form.
By Corollary~\ref{cor_WeightSol}, we could extract from $\M[\isell] B$ a $\iQ(X,Y) \in \Mod[\isell] M$ with minimal $(1,k-1)$-weighted degree.
Since it must satisfy the interpolation conditions of Theorem~\ref{thm_GSproblem}, and since the weighted degree is minimal among such polynomials, it must also satisfy the degree constraints for $\itau = \tau(\isell)$.
By that theorem any codeword with distance at most $\itau$ from $\vec r$ would then be represented by a root of $\iQ(X,Y)$.

Algorithm~\ref{alg_multitrial} is a generalisation and formalisation of this method.
For a given $\RS n k$ code, one chooses ultimate parameters $(s, \ell, \tau)$ being a permissible triple with $s \leq \ell$.
One also chooses a list of micro-steps and chooses after which micro-steps to attempt decoding; these choices are represented by a list of $\var S_1, \var S_2$ and $\var{Root}$ elements.
This list must contain exactly $s-\ell$ $\var S_1$-elements of  and $s-1$ $\var S_2$-elements, as it begins by computing a basis for $\Mod[1,1] M$ and will end with a basis for $\Mod[\sell] M$.
If there is a $\var{Root}$ element in the list, the algorithm finds all codewords with distance at most $\itau = \tau(\isell)$ from $\vec r$; if this list is non-empty, the computation breaks and the list is returned.

The algorithm calls sub-functions which we explain informally: $\var{MicroStep1}$ and $\var{MicroStep2}$ will take $\isell$ and a basis in weak Popov form for $\Mod[\isell] M$ and return a basis in weak Popov form for $\Mod[\is,\iell+1] M$ respectively $\Mod[\is+1,\iell+1] M$; more detailed descriptions for these are given in Subsections~\ref{ssec_steptypeone} and \ref{ssec_steptypetwo}.
$\var{MinimalWeightedRow}$ finds a polynomial of minimal $(1,k-1)$-weighted degree in $\M[\isell]{M}$ given a basis in weak Popov form (Corollary~\ref{cor_WeightSol}).
Finally, $\var{RootFinding}(Q, \tau)$ returns all $Y$-roots of $Q(X,Y)$ of degree less than $k$ and whose corresponding codeword has distance at most $\tau$ from the received word $\vec r$.

\printalgos{\caption{Multi-Trial Guruswami--Sudan Decoding}
\label{alg_multitrial} 
\DontPrintSemicolon 
\SetAlgoVlined
\LinesNumbered
\SetKwInput{KwPre}{{Preprocessing}}
\SetKwInput{KwIn}{{Input}}
\SetKwInput{KwOut}{{Output}}
\BlankLine
\KwIn{A $\RS{n}{k}$ code and the received vector $\vec r = (r_0, \dots, r_{n-1})$\;
A permissible triple $(s,\ell,\tau)$\;
A list \var{C} with elements in $\{\var S_1, \var S_2, \var{Root} \}$ with $s-1$  instances of $\var S_2$, $\ell-s$ instances of $\var S_1$\;
}
\BlankLine
\KwPre{Calculate $r_i^{\prime} = r_i/w_i$ for all $i=0,\dots,n-1$\;
Construct $\M[1,1]{A}$, and compute $\M[1,1]{B}$ from $\M[1,1]{A} \M[1]{W}$ using Alekhnovich's algorithm\;
Initial parameters $(\is, \iell) \leftarrow (1, 1)$ \;}
\BlankLine
\For{each $\var c$ in $\var{C}$}{
	\If{$c = \var S_1$ \label{alg_line_firstif}}{
		$\M[\is,\iell+1]{B} \leftarrow \var{MicroStep1}(\is, \iell, \M[\isell]{B})$\;
		$(\isell) \leftarrow (\is,\iell+1)$\;
	}
	\If{$c = \var S_2$ \label{alg_line_secondif}} {
		$\M[\is+1,\iell+1]{B} \leftarrow \var{MicroStep2}(\is, \iell, \M[\isell]{B})$\;
		$(\isell) \leftarrow (\is+1,\iell+1)$\;
	}
	\If{$c = \var{Root}$ \label{alg_line_thirdif}} {
          $Q(X,Y) \leftarrow \var{MinimalWeightedRow}(\M[\isell]{B})$ \label{alg_line_extractpoly}\;
          \If{$\var{RootFinding}(Q(X,Y), \tau(\isell)) \neq \emptyset$}{
            \Return this list\;
          }
        }
}
}

Algorithm~\ref{alg_multitrial} has a large amount of flexibility in the choice of the list $\var C$, but since we can only perform micro-steps of type I and II, there are choices of $s$ and $\ell$ we can never reach, or some which we cannot reach if we first wish to reach an earlier $s$ and $\ell$.
We can never reach $s > \ell$, but as mentioned in Section~\ref{sec_prelim}, such a choice never makes sense. 
It also seems to be the case that succession of sensibly chosen parameters can always be reached by micro-steps of type I and II.
That is, if we first wish to attempt decoding at some radius $\tau_1$ and thereafter continue to $\tau_2 > \tau_1$ in case of failure, the minimal possible $s_1, \ell_1$ and $s_2, \ell_2$ such that $(s_1,\ell_1,\tau_1)$ respectively $(s_2,\ell_2,\tau_2)$ are permissible will satisfy $0 \leq s_2-s_1 \leq \ell_2-\ell_1$.
However, we have yet to formalise and prove such a statement.

\begin{inJournal}
Ultimately we are interested in attempting to decode at various radii, and so the following corollary describes that we are in fact not limited in this respect:
\begin{cor}
  Let $(s_0,\ell_0,\tau_0), \ldots, (s_h,\ell_h,\tau_h)$ be a sequence of permissible triples such that $\tau_i < \tau_{i+1}$ and $(s_0,\ell_0) = (1, 1)$.
  Then there exists $s'_0,\ldots,s'_h$ and $\ell'_0,\ldots,\ell'_h$ such that $(s'_i,\ell'_i,\tau_i)$ is permissible for all $i$ and where $0 \leq s'_{i+1} - s'_i \leq \ell_{i+1} - \ell_i$. Furthermore, $s'_i \leq s_i$ and $\ell'_i \leq \ell_i$.
\end{cor}
\begin{proof}
  We prove the corollary by transforming the sequence of permissible triples in three passes into one which satisfies the conditions.

  In the first pass, let $i$ be the first index such that $s_i > s_{i+1}$, if any.
  By Lemma TODO, we can replace $(s_i,\ell_i,\tau_i)$ with the permissible $(s_{i+1},\ell_i,\tau_i)$.
  Continue like this with the rest of the list.

  In the second pass, we remove all occurrences where $\ell_i > \ell_{i+1}$ in a manner similar to before.
  
  Now the list satisfies that $s_{i+1} \geq s_i$ and $\ell_{i+1} \geq \ell_i$, so for the final pass, let $i$ be the first index such that $s_{i+1}-s_i > \ell_{i+1} - \ell_i$.
  By Lemma~\ref{lem_permSdiff}, we can replace $(s_i,\ell_i,\tau_i)$ by the permissible $(s_i + (\ell_{i+1}-\ell_i), \ell_i, \tau_i)$.
  \qed
\end{proof}
Note that for a sequence of parameters resulting from the above corollary, it will always be possible to choose a series of micro-steps of type I and II.
Computationally, this sequence will also not be worse than the original since it outputs only $s'_i$ and $\ell'_i$ which not larger than the original.
\end{inJournal}

In the following two subsections we explain the details of the micro-steps.
In Section \ref{ssec_complanalysis}, we discuss the complexity of the method and how the choice of $\var C$ influence this.

\subsection{Micro-Step Type I: $(s,\ell) \mapsto (s, \ell+1)$} \label{ssec_steptypeone}

\begin{lem}
  If $B\T 0(X,Y), \ldots, B\T \ell(X,Y)$ is a basis of $\Mod[\sell]{M}$, then the following is a basis of $\Mod[\sellI] M$:
  \[
  B\T 0(X,Y),\ \ldots\ ,\ B\T \ell(X,Y), Y^{\ell-s+1}(Y-R(X))^s
  \]
\end{lem}
\begin{proof}
  In the basis of $\Mod[\sellI]{M}$ given in Theorem~\ref{thm_Mbasis}, the first $\ell+1$ generators are the generators of $\Mod[\sell]{M}$.
  Thus all of these can be described by any basis of $\Mod[\sellI]{M}$.
  The last remaining generator is exactly $Y^{\ell-s+1}(Y-R(X))^s$.
  \qed
\end{proof}

In particular, the above lemma holds for a basis of $\Mod[\sellI] M$ in weak Popov form, represented by a matrix $\M[\sell]{B}$.
The following matrix thus represents a basis of $\Mod[\sellI]{M}$:
\begin{equation}\label{eqn_stepI}
  \M[\sellI] C\stepI=
  \left[\begin{array}{r}
      \begin{array}{@{}c|c@{}}
      \\[-0.2cm]
      \makebox[4cm][c]{$\M[\sell]{B} $}
      & 
      \makebox[1.5em][r]{$\vec 0^T$}
      \\[-0.2cm] \
    \end{array}
    \\ \hline \\[-.3cm]
    \begin{matrix}
      0 & \ldots & 0 & (-R)^s & \binom s 1 (-R)^{s-1} & \ldots & 1
    \end{matrix}
  \end{array}\right].
\end{equation}
\begin{lem} 
  \label{lem_orthoI}
  $\OD{\M[\sellI]{C}\stepI \M[\ell+1]{W}} = s(\deg R - k + 1) \leq s(n-k)$.
\end{lem}
\begin{proof}
  We calculate the two quantities $\det(\M[\sellI] C\stepI \M[\ell+1]{W})$ and $\deg(\M[\sellI] C\stepI \M[\ell+1]{W})$.
  It is easy to see that
  \[
  \det(\M[\sellI] C\stepI \M[\ell+1]{W})
  = \det \M[\sell]{B} \det \M[\ell+1]{W}
  = \det \M[\sell]{B} \det \M[\ell]{W} X^{(\ell+1)(k-1)}.
  \]
  For the row-degree, it is clearly $\deg(\M[\sell]B\M[\ell] W)$ plus the row-degree of the last row.
  If and only if the received word is not a codeword then $\deg R \geq k$, then the leading term of the last row must be $(-R)^sX^{(\ell+1-s)(k-1)}$.
  Thus, we get
  \begin{align*}
    \OD{\M[\sellI]C\stepI \M[\ell+1] W}
    &= \big( \deg(\M[\sell] B \M[\ell] W) + s \deg R + (\ell+1-s)(k-1) \big)
    \\ &\qquad - \big( \deg\det(\M[\sell] B \M[\ell] W) + (\ell+1)(k-1) \big)
    \\ &= s(\deg R - k + 1),
  \end{align*}
  where the last step follows from Lemma~\ref{lem_popovreduces} as $\M[\sell]B \M[\ell] W$ is in weak Popov form.
  \qed
\end{proof}
\begin{cor}
  \label{cor_ComplMSI}
  The complexity of $\var{MicroStep1}(s,\ell,\M[\sell]B)$ is $O(\ell^3 s n \log^2 n \log\log n)$.
\end{cor}
\begin{proof}
  Follows by Lemma~\ref{lem_alekcompl}. Since $s \in O(n^2)$ we can leave out the $s$ in $\log$-terms.
  \qed
\end{proof}

\subsection{Micro-Step Type II: $(s,\ell) \mapsto (s+1, \ell+1)$} \label{ssec_steptypetwo}

\begin{lem}
  If $B\T 0(X,Y), \ldots, B\T \ell(X,Y)$ is a basis of $\Mod[\sell]{M}$, then the
  following is a basis of $\Mod[\sellII]M$:
  \[
    G^{s+1}(X),\ B\T 0(X,Y)(Y-R(X)),\ \ldots\ ,\ B\T \ell(X,Y)(Y-R(X)).
  \]
\end{lem}
\begin{proof}
  Denote by $P_{\sell}\T 0(X,Y),\ldots,P_{\sell}\T \ell(X,Y)$ the basis of $\Mod[\sell]M$ as given in Theorem~\ref{thm_Mbasis}, and by $P_{\sellII}\T 0(X,Y),\ldots,P_{\sellII}\T{\ell+1}(X,Y)$ the basis of $\Mod[\sellII] M$.
  Then observe that for $t > 0$, we have $P_{\sellII}\T t = P_{\sell}\T{t-1}(Y-R(X))$.
  Since the $B\T i(X,Y)$ form a basis of $\Mod[\sell]{M}$, each $P_{\sell}\T t$ is expressible as an $\Fqx$-combination of these, and thus for $t>0$, $P_{\sellII}\T t$ is expressible as an $\Fqx$-combination of the $B\T i(X,Y)(Y-R(X))$.
  Remaining is then only $P_{\sellII}\T 0(X,Y) = G^{s+1}(X)$.
  \qed
\end{proof}
As before, we can use the above with the basis $\M[\sell] B$ of $\Mod[\sell] M$ in weak Popov form, found in the previous iteration of our algorithm.
Remembering that multiplying by $Y$ translates to shifting one column to the right in the matrix representation, the following matrix thus represents a basis of $\Mod[\sellII]{M}$:
\begin{equation}\label{eqn_stepII}
  \M[\sellII]C\stepII =
    \left[\begin{array}{@{}c|c@{}}
        G^{s+1} & \vec 0 \\\hline\\[-.2cm]
        \vec 0^T & \makebox[1cm][c]{$\vec 0$}
          \\[0.1cm]
    \end{array}\right]
    +
    \left[\begin{array}{@{}c|c@{}}
        0 &  \vec 0 \\\hline\\[-.2cm]
        \vec 0^T & \makebox[1cm][c]{$\M[\sell]{B}$}
          \\[0.1cm]
    \end{array}\right]
    - R\cdot
    \left[\begin{array}{@{}c|c@{}}
        \vec 0 & 0 \\\hline\\[-.2cm]
        \makebox[1cm][c]{$\M[\sell]{B}$} & \vec 0^T
          \\[0.1cm]
    \end{array}\right].
\end{equation}
\begin{lem} 
  \label{lem:Cortho}
  $\OD{\M[\sellII]C\stepII \M[\ell+1] W}
  = (\ell+1)(\deg R - k + 1)
  \leq (\ell+1)(n-k)$.
\end{lem}
\begin{proof}
  We compute $\deg(\M[\sellII]C\stepII \M[\ell+1] W)$ and $\deg \det(\M[\sellII]C\stepII \M[\ell+1] W)$.
  For the former, obviously the first row has degree $(s+1)n$.
  Let $\vec{b}_i$ denote the $i$th row of $\M[\sell] B$ and $\vec{b}'_i$ denote the $i$th row of $\M[\sell] B \M[\ell] W$.
  The $(i+1)$th row of $\M[\sellII]C\stepII \M[\ell+1]W$ has the form
  \[
    \big[(0 \mid \vec{b}_i) - R(\vec{b}_i \mid 0)\big]\M[\ell+1] W
      = (0 \mid \vec{b}'_i)X^{k-1} - R(\vec{b}'_i \mid 0).
  \]
  If and only if the received word is not a codeword, then $\deg R \geq k$.
  In this case, the leading term of $R \vec{b}'_i$ must have greater degree than any term in $X^{k-1}\vec{b}'_i$.
  Thus the degree of the above row is $\deg R + \deg \vec{b}'_i$.
  Summing up we get
  \begin{align*}
    \deg \M[\sellII]C\stepII &= (s+1)n + \sum_{i=0}^\ell \deg R + \deg \vec{b}'_i \\
    &= (s+1)n + (\ell+1)\deg R + \deg(\M[\sell]{B} \M[\ell] W).
  \end{align*}
  For the determinant, observe that
  \begin{align*}
    \det(\M[\sellII]C\stepII \M[\ell+1] W)
    &= \det(\M[\sellII]C\stepII) \det(\M[\ell+1] W) \\
    &= G^{s+1} \det \tilde{\M B} \det \M[\ell] W X^{(\ell+1)(k-1)},
  \end{align*}
  where $\tilde{\M B} = \M[\sell]{B} - R \left[ {\M[\sell] {\grave B}} \ \big|\  \vec 0^T \right]$ and $\M[\sell]{\grave B}$ is all but the zeroth column of $\M[\sell] B$.
  This means $\tilde{\M B}$ can be obtained by starting from $\M[\sell] B$ and iteratively adding the $(j+1)$th column of $\M[\sell] B$ scaled by $R(X)$ to the $j$th column, with $j$ starting from $0$ up to $\ell-1$.
  Since each of these will add a scaled version of an existing column in the matrix, this does not change the determinant.
  Thus, $\det \tilde{\M B} = \det \M[\sell] B$.
  But then $\det \tilde{\M B} \det \M[\ell] W = \det(\M[\sell] B \M[\ell] W)$ and so $\deg(\det \tilde{\M B}\det \M[\ell] W) = \deg(\M[\sell] B \M[\ell] W)$ by Lemma~\ref{lem_popovreduces} since $\M[\sell] B \M[\ell] W$ is in weak Popov form.
  Thus we get
  \[
    \deg \det(\M[\sellII]C\stepII \M[\ell+1] W) =
      (s+1)n + \deg(\M[\sell] B \M[\ell] W) + (\ell+1)(k-1).
  \]
  The lemma follows from the difference of the two calculated quantities.
  \qed
\end{proof}
\begin{cor}
  \label{cor_ComplMSII}
  The complexity of $\var{MicroStep2}(s,\ell,\M[\sell]B)$ is $O(\ell^4 n \log^2 n \log\log n)$.
\end{cor}

\begin{inJournal}
\begin{exa}
  \label{ex_GRSrevisited}
\todo{/Johan: Is this example not way too long? I understand the desire for an example but perhaps this is going into too much detail.
One could remove all comparisons with A-matrices, and one could remove the structural form of C[1,2]. The example will still be quite long. It will not get shorter if we also want to include explicit Q-polynomials.
Otherwise, I completely agree that it should be for a particular example.}
We consider again the \RS{16}{4} code of Example~\ref{ex_rs164param}.
We will consider $\tau = 8$ as our goal and therefore the permissible triple $ (2,4,8)$ as the maximal parameter choice. 
To maximise decoding radius during the procedure, we could choose the following
sequence for $(s,\ell,\tau)$:
\todo{/Johan: A command list should be specified here}
\begin{equation*}
(1,1,6) \rightarrow (1,2,7) \rightarrow (2,3,7) \rightarrow (2,4,8).
\end{equation*}
That is, our micro-steps will by of type I, II, and I in that order.
Let us depict the matrices in the decoding process of these three steps. For
some matrix $\M{B}=\|b_{i,j}\|$ over $\Fqx$, we will write $\M{B} \preceq \M{\hat B}$ for some matrix
$\M{\hat B}$ over $\ZZ$ of the same dimensions, if $\deg b[i,j] \leq \hat b_{i,j}$ for all $i,j$. Let $\bot = -1$, i.e., in use with $\preceq$, it denotes that only the zero polynomial is allowed at this place.

To begin with, we have:
\begin{align*}
\M[1,1]{A} = \left(\begin{array}{rr}
                G & 0 \\
                -R & 1
              \end{array}\right) \quad
\M[1,1]{A} \preceq \left(\begin{array}{rr}
      16 & \bot \\
      15 & 0
  \end{array}\right)  \quad
\M[1,1]{A} \M[1]{W} \preceq \left(\begin{array}{rr}
      16 & \bot \\
      15 & 3
  \end{array}\right)
\end{align*}
according to~\eqref{eq_BasisOfModul} and~\eqref{eq_DiagWeight}.
We then apply Alekhnovich's algorithm on $\M[1,1]{A} \M[1]W$ to obtain $\M[1,1]{B} \M[1]{W}$ which is in weak Popov form.
From this we can easily scale down the columns again to obtain $\M[1,1]{B}$.

We get $\M[1,1]{A} \M[1]{W} = \deg(\M[1,1]{A}\M[1]{W}) - \deg\det(\M[1,1]{A} \M[1]{W}) = (2n-1) -n = 15$, so we expect to apply at most $r(\OD{\M[1,1]{A}}+r) = 34$ simple transformations before reaching $\M[1,1]{B} \M[1]{W}$. 

\todo{A2J: Here the complexity should be bounded without using the number of simple transformations. /Johan: I don't think we should say anything about complexity here, since all our complexity assertations stem from the lemma on Alekhnovich we just pulled from his paper. Instead, just remove all such comments and comparisons with A matrices (except of course A[1,1]}

We know that there must be a $Q(X,Y) \in \Mod[1,1]{M}$ and so a row of less than this degree must be in $\M[1,1]{B}$. Furthermore, since it is in weak Popov form, each row will have different leading positions. A good guess at its approximate degrees would then be:
\begin{align*}
  \M[1,1]{B} \preceq \left(\begin{array}{rr}
    9 & 6 \\
    10 & 6
  \end{array}\right) \quad \quad
  \M[1,1]{B}\M[1]{W} \preceq \left(\begin{array}{rr}
    9 & 9 \\
    10 & 9
  \end{array}\right)
\end{align*}
\todo{A2J: For a specific example we should give here the exact degrees... /Johan: Yes and therefore shorten the blabbering :-)}

We would most likely want to factor now, so we pick a row in $\M[1,1]{B}$ which
has weighted degree less than $10$, say $(\yC Q 0(X), \yC Q 1(X))$, and construct $Q(X,Y) = \yC Q 0(X) + Y\yC Q 1(X)$. 

\todo{A2J: For a specific example we should give $Q(X,Y)$ explicitly and show that root-finding for $\tau=6$ fails.}

In case this fails, we move to the next step, i.e.,  $(s,\ell, \tau) = (1,2,7)$. From~\eqref{eqn_stepI}, we get
\begin{align*}
  \M[1,2]C\stepI  =
    \left[\begin{array}{r}
            \begin{array}{@{}c|c@{}}
        \makebox[1cm][c]{$\M[1,1]{B}$}
        & 
        \begin{matrix} 0 \\ 0 \end{matrix}
      \end{array}
      \\ \hline \\[-.4cm]
      \begin{matrix}
        0 & -R & 1
      \end{matrix}
    \end{array}\right]
\end{align*}
and so
\begin{align*}
      \M[1,2]C\stepI \preceq \left(\begin{matrix}
      9 & 6 & \bot \\
      10 & 6 & \bot \\
      \bot & 15 & 0
    \end{matrix}\right) \quad \quad
      \M[1,2]C\stepI\M[2]{W} \preceq \left(\begin{matrix}
      9 & 9 & \bot \\
      10 & 9 & \bot \\
      \bot & 18 & 6
    \end{matrix}\right)
\end{align*}
Row-reduction essentially ``balances'' the row-degrees such that they are all roughly of the same size, and the complexity for performing this reduction is in the size of the ``unbalancedness''. It is clear that the above is more
balanced in row-degrees than finding $\M[1,2]{B}$ using the straightforward basis
$\M[1,2]{A}$:
\begin{align*}
    \M[1,2]{A} \preceq \left(\begin{matrix}
    16 & \bot & \bot \\
    15 & 0 & \bot \\
    \bot & 15 & 0
  \end{matrix}\right) \quad \quad
    \M[1,2]{A} \M[2]{W} \preceq \left(\begin{matrix}
    16 & \bot & \bot \\
    15 & 3 & \bot \\
    \bot & 18 & 6
  \end{matrix}\right)
\end{align*}
Running Alekhnovich's on $\M[1,2]C\stepI\M[2]{W}$, we obtain $\M[1,2]{B} \M[2]{W}$. By Lemma~\ref{lem_orthoI} and Corollary~\ref{cor_ComplMSI}, this should happen after at most $r(\OD{D_{1,2}\M[2]{W}}+r) = 3(s(n-k)+3) = 45$ simple transformations.
\todo{A2J: Here the complexity should be bounded without using the number of simple transformations. /Johan: The whole block above should be removed since it was written to convice you}

By reasoning as before, we know that at least one row in this matrix must have
row-degree less than $9$. Thus we have
\begin{align*}
    \M[1,2]{B} \preceq \left(\begin{matrix}
    8 & 5 & 2 \\
    9 & 5 & 2 \\
    9 & 6 & 2
  \end{matrix}\right) \quad \quad
    \M[1,2]{B} \M[2]{W} \preceq \left(\begin{matrix}
    8 & 8 & 8 \\
    9 & 8 & 8 \\
    9 & 9 & 8
  \end{matrix}\right).
\end{align*}

\todo{A2J: Factoring in this step.}

The next permissible triple is $(2,4,8)$ \todo{/Johan: No its (2,3,7)?} and the next micro-step is of type II so by~\eqref{eqn_stepII}, we construct:
\begin{align*}
  \M[2,3]C\stepII \preceq \left(\begin{matrix}
      32 & \bot & \bot & \bot \\
      23 & 20 & 17 & 2 \\
      24 & 20 & 17 & 2 \\
      23 & 21 & 17 & 2
  \end{matrix}\right) \quad \quad
  \M[2,3]C\stepII \M[3]{W} \preceq \left(\begin{matrix}
      32 & \bot & \bot & \bot \\
      23 & 23 & 23 & 11 \\
      24 & 23 & 23 & 11 \\
      23 & 24 & 23 & 11
  \end{matrix}\right)
\end{align*}
Once again, this is clearly much more ``balanced'' than the straightforward basis $\M[2,3]{A}$:
\begin{align*}
    \M[2,3]{A} = \left(\begin{matrix}
      G^{2} & 0 & 0 & 0 \\
      -G R & G & 0 & 0 \\
      R^{2} & -2 \, R & 1 & 0 \\
      0 & R^{2} & -2 \, R & 1
  \end{matrix}\right) \quad \quad
    \M[2,3]{A} \preceq \left(\begin{matrix}
      32 & \bot & \bot & \bot \\
      31 & 16 & \bot & \bot \\
      30 & 15 & 0 & \bot \\
      \bot & 30 & 15 & 0
  \end{matrix}\right) 
\end{align*}
By Lemma \ref{lem:Cortho}, the number of simple transformations that Alekhnovich's algorithm will apply will be upper bounded by
$r(\OD{C_{2,3}W_3}+r) = 4(3(n-k) + 4) = 160$. 

\todo{A2J: Complexity without simple transformations.}

After row-reduction, we obtain $\M[2,3]{B} \M[3]{W}$, and again know that there will be a row with degree less than $18$, and so:
\begin{align*}
    \M[2,3]{B} \preceq \left(\begin{matrix}
      17 & 14 & 11 & 8 \\
      18 & 14 & 11 & 8 \\
      18 & 15 & 11 & 8 \\
      18 & 15 & 12 & 8
  \end{matrix}\right) \quad \quad
    \M[2,3]{B} \M[3]{W} \preceq \left(\begin{matrix}
      17 & 17 & 17 & 17 \\
      18 & 17 & 17 & 17 \\
      18 & 18 & 17 & 17 \\
      18 & 18 & 18 & 17
  \end{matrix}\right)
\end{align*}
Since the decoding radius did not increase, it never makes sense to factor here. Therefore, we continue with the last micro-step of type II and get:
\begin{align*}
    \M[2,4]C\stepI \preceq \left(\begin{matrix}
      17 & 14 & 11 & 8 & \bot \\
      18 & 14 & 11 & 8 & \bot \\
      17 & 15 & 11 & 8 & \bot \\
      17 & 14 & 12 & 8 & \bot \\
      \bot & \bot & 30 & 15 & 0
  \end{matrix}\right) & \quad \quad
    \M[2,4]C\stepI \M[4]{W} \preceq \left(\begin{matrix}
      17 & 17 & 17 & 17 & \bot \\
      18 & 17 & 17 & 17 & \bot \\
      17 & 18 & 17 & 17 & \bot \\
      17 & 17 & 18 & 17 & \bot \\
      \bot & \bot & 36 & 24 & 12
  \end{matrix}\right) \\
  \M[2,4]{B} \preceq \left(\begin{matrix}
      15 & 12 & 9 & 6 & 3 \\
      16 & 12 & 9 & 6 & 3 \\
      16 & 13 & 9 & 6 & 3 \\
      16 & 13 & 10 & 6 & 3 \\
      16 & 13 & 10 & 7 & 3
  \end{matrix}\right) & \quad \quad
  \M[2,4]{B} \M[4]{W} \preceq \left(\begin{matrix}
      15 & 15 & 15 & 15 & 15 \\
      16 & 15 & 15 & 15 & 15 \\
      16 & 16 & 15 & 15 & 15 \\
      16 & 16 & 16 & 15 & 15 \\
      16 & 16 & 16 & 16 & 15
  \end{matrix}\right).
\end{align*}
Reducing $\M[2,4]C\stepI \M[4]{W}$ should take at most $r(\OD{D_{2,4}\M[4]{W}}+r) = 5(2(n-k)+
5) = 145$ simple transformations. And compare again with
\begin{align*}
    \M[2,4]{A} = \left(\begin{matrix}
      G^{2} & 0 & 0 & 0 & 0 \\
      -G R & G & 0 & 0 & 0 \\
      R^{2} & -2 \, R & 1 & 0 & 0 \\
      0 & R^{2} & -2 \, R & 1 & 0 \\
      0 & 0 & R^{2} & -2 \, R & 1
  \end{matrix}\right) \quad \quad
  \M[2,4]{A} \preceq \left(\begin{matrix}
      32 & \bot & \bot & \bot & \bot \\
      31 & 16 & \bot & \bot & \bot \\
      30 & 15 & 0 & \bot & \bot \\
      \bot & 30 & 15 & 0 & \bot \\
      \bot & \bot & 30 & 15 & 0
  \end{matrix}\right).
\end{align*}
is much...

\todo{need to be finished with a nice sentence...}
\end{exa}
\end{inJournal}

\subsection{Complexity Analysis} \label{ssec_complanalysis}

Using the estimates of the two preceding subsections, we can make a rather precise worst-case asymptotic complexity analysis of our multi-trial decoder.
The average running time will depend on the exact choice of $\var C$ but we will see that the worst-case complexity will not.
First, it is necessary to know the complexity of performing a root-finding attempt.
\begin{lem}[Complexity of Root-Finding] \label{lem_rootCompl}
  Given a polynomial $Q(X,Y) \in \Fqx[Y]$ of $Y$-degree at most $\ell$ and $X$-degree at most $N$, there exists an algorithm to find all $\Fqx$-roots of complexity $O\big(\ell^2N\log^2 N\log\log N\big)$, assuming $\ell, q \in O(N)$.
\end{lem}
\begin{proof}
  We employ the Roth--Ruckenstein~\cite{rothRuckenstein00} root-finding algorithm together with the divide-and-conquer speed-up by Alekhnovich~\cite{alekhnovich05}.
  The complexity analysis in~\cite{alekhnovich05} needs to be slightly improved to yield the above, but see \cite{nielsen12} for easy amendments.
\end{proof}

\begin{thm}[Complexity of Algorithm~\ref{alg_multitrial}] \label{thm_ComplAlg}
  For a given $\RS n k$ code, as well as a given list of steps $\var C$ for Algorithm~\ref{alg_multitrial} with ultimate parameters $(s, \ell, \tau)$, the algorithm has worst-case complexity $O(\ell^4 s n \log^2 n \log\log n)$, assuming $q \in O(n)$.
\end{thm}
\todo{A2J: I would prefer some "names" for the Lemma and the Theorem, but this is a old story... /Johan: I don't care much but prefer not to have them. Especially when the names are not very obvious or reference to well-known ones. I have removed many of yours already in the name of compacting. I don't think we should introduce them now --- we're still pressed on space.}

\begin{proof}
  The worst-case complexity corresponds to the case that we do not break early but run through the entire list $\var C$.
  Precomputing $\M[\sell] A$ using Lagrangian interpolation can be performed in $O(n\log^2 n\log\log n)$, see e.g. \cite[p. 235]{gathen}, and reducing to $\M[\sell] B$ is in the same complexity by Lemma~\ref{lem_alekcompl}.
  
  Now, $\var C$ must contain exactly $\ell-s$ $\var S_1$-elements and $s-1$ $\var S_2$-elements.
  The complexities given in Corollaries~\ref{cor_ComplMSI} and \ref{cor_ComplMSII} for some intermediate $\isell$ can be relaxed to $s$ and $\ell$.
  Performing $O(\ell)$ micro-steps of type I and $O(s)$ of type II is therefore in $O(\ell^4 s n \log^2 n\log\log n)$.

  It only remains to count the root-finding steps.
  Obviously, it never makes sense to have two $\var{Root}$ after each other in $\var C$, so after removing such possible duplicates, there can be at most $\ell$ elements $\var{Root}$. 
  When we perform root-finding for intermediate $\isell$, we do so on a polynomial in $\Mod[\isell] M$ of minimal weighted degree, and by the definition of $\Mod[\isell] M$ as well as Theorem \ref{thm_GSproblem}, this weighted degree will be less than $\is(n-\itau) < sn$.
  Thus we can apply Lemma \ref{lem_rootCompl} with $N = sn$.
  \qed
\end{proof}
The worst-case complexity of our algorithm is equal to the average-case complexity of the Beelen--Brander~\cite{beelenBrander} list decoder. However, Theorem~\ref{thm_ComplAlg} shows that we can choose as many intermediate decoding attempts as we would like without changing the worst-case complexity.
One could therefore choose to perform a decoding attempt just after computing $\M[1,1] B$ as well as every time the decoding radius has increased.
The result would be a decoding algorithm finding all \emph{closest} codewords within some ultimate radius $\tau$.
If one is working in a decoding model where such a list suffices, our algorithm will thus have much better average-case complexity since fewer errors occur much more frequently than many.

\section{Conclusion} \label{sec_concl}
An iterative interpolation procedure for list decoding GRS codes based on Alekhnovich's module minimisation was proposed and shown to have the same worst-case complexity as Beelen and Brander's \cite{beelenBrander}.
We showed how the target module used in Beelen--Brander can be minimised in a progressive manner, starting with a small module and systematically enlarging it, performing module minimisation in each step.
The procedure takes advantage of a new, slightly more fine-grained complexity analysis of Alekhnovich's algorithm, which implies that each of the module refinement steps will run fast.

The main advantage of the algorithm is its granularity which makes it possible to perform fast multi-trial decoding: we attempt decoding for progressively larger decoding radii, and therefore find the list of codewords closest to the received.
This is done without a penalty in the worst case but with an obvious benefit in the average case.

\subsection*{Acknowledgement}
The authors thank Daniel Augot for fruitful discussions. 
This work has been supported by German Research Council “Deutsche Forschungsgemeinschaft” (DFG) under grant BO 867/22-1. Johan S. R. Nielsen also gratefully acknowledges the support from The Otto M\o{}nsted Foundation and the Idella Foundation.

\bibliographystyle{spmpsci}
\bibliography{../../../bibtex}

\end{document}